\documentclass[11pt]{article}
\usepackage[utf8]{inputenc}

\advance \oddsidemargin by -0.8in

\usepackage{verbatim}
\usepackage{float}
\usepackage{amsmath}
\usepackage{amsthm}

\makeatletter

\floatstyle{ruled}
\newfloat{algorithm}{tbp}{loa}
\providecommand{\algorithmname}{Algorithm}
\floatname{algorithm}{\protect\algorithmname}

\theoremstyle{plain}
\usepackage{xcolor}
\usepackage{nameref}
\definecolor{ForestGreen}{rgb}{0.1333,0.5451,0.1333}
\definecolor{DarkRed}{rgb}{0.8,0,0}
\definecolor{Red}{rgb}{1,0,0}
\usepackage[linktocpage=true,
pagebackref=true,colorlinks,
linkcolor=DarkRed,citecolor=ForestGreen,
bookmarks,bookmarksopen,bookmarksnumbered]
{hyperref}

\usepackage{cleveref}

\newtheorem{theorem}{Theorem}[section]
\newtheorem{thm}[theorem]{Theorem}

\newtheorem{lem}[theorem]{Lemma}

\newtheorem{cor}[theorem]{Corollary}

\newtheorem{defn}[theorem]{Definition}

\makeatother

\begin{document}
\global\long\def\vol{\mathrm{vol}}
\global\long\def\X{\mathcal{X}}
\global\long\def\Ohat{\widehat{O}}
\global\long\def\trim{\textsc{trim}}
\global\long\def\shave{\textsc{shave}}
\global\long\def\expander{\textsc{expander}}
\global\long\def\poly{\mathrm{poly}}
\global\long\def\polylog{\mathrm{polylog}}

\title{A Simple Deterministic Algorithm for Edge Connectivity}

\author{Thatchaphol Saranurak}
\maketitle
\begin{abstract}
We show a deterministic algorithm for computing edge connectivity
of a simple graph with $m$ edges in $m^{1+o(1)}$ time. Although
the fastest deterministic algorithm by Henzinger, Rao, and Wang {[}SODA'17{]}
has a faster running time of $O(m\log^{2}m\log\log m)$, we believe
that our algorithm is conceptually simpler. The key tool for this
simplication is the \emph{expander decomposition}. We exploit it in
a very straightforward way compared to how it has been previously
used in the literature.
\end{abstract}

\section{Introduction}

Edge connectivity is a fundamental measure for robustness of graphs.
Given an undirected graph $G=(V,E)$ with $n$ vertices and $m$ edges,
the edge connectivity $\lambda$ of $G$ is the minimum number of
edges whose deletion from $G$ disconnects $G$. These edges correspond
to a \emph{(global) minimum cut} $(C,V\setminus C)$ where the number
of edges crossing the cut is $|E(C,V\setminus C)|=\lambda$. Numerous
algorithms for computing edge connectivity have been discovered and
are based on various fascinating techniques, including exact max flow
computation \cite{ford1962flows,HaoO94,LiP_submodular}, maximum adjacency
ordering \cite{NagamochiI92_mincut,StoerW97,frank2009edge}, random
contraction \cite{Karger93,KargerS96}, arborescence packing \cite{Gabow95,GabowM98},
and greedy tree packing and minimum cuts that 2-respect a tree \cite{Karger00,BhardwajLS20,GawrychowskiMW20,MukhopadhyayN20,GawrychowskiMW20_note}.
All of these techniques also extend to weighted graphs where we need to find a cut with minimum total edge weight crossing the
cut. 

Quite recently, Kawarabayashi and Thorup \cite{KawarabayashiT19}
showed a novel technique for computing edge connectivity of \emph{simple}
unweighted graphs (i.e.~graphs with no parallel edges) in $O(m\log^{12}n)$
time deterministically. This technique leads to the fastest deterministic
algorithm with $O(m\log^{2}n\log\log n)$ time by Henzinger, Rao,
and Wang \cite{HenzingerRW17}, and the fastest randomized algorithm
with running time $\min\{O(m\log n),O(m+n\log^{3}n)\}$ with high
probability by Ghaffari, Nowicki, and Thorup \cite{Ghaffari0T20}.
The state-of-the-art algorithms for non-simple graphs have slower
running times. 

The core idea in this line of work is a new contraction technique
that preserves all \emph{non-trivial} minimum cuts. Recall that \emph{trivial}
cuts $(C,V\setminus C)$ are cuts where $\min\{|C|,|V\setminus C|\}=1$.
Although the algorithm by \cite{Ghaffari0T20} already gave a simple
implementation of this idea using randomization, all deterministic
algorithms for finding such a contraction are still quite complicated.
For example, they require intricate analysis of personalized PageRank
\cite{KawarabayashiT19} and local flow technique \cite{HenzingerRW17}
and a non-trivial way for combining all algorithmic tools together. 

In this paper, we observe that such a contraction follows almost immediately
from the \emph{expander decomposition} introduced in \cite{KannanVV04}.
Although the best-known implementation of expander decomposition itself
is not yet very simple \cite{SaranurakW19,ChuzhoyGLNPS20}, given
it as a black-box, our algorithm can be described in only few steps
and we believe that it offers a \emph{conceptual} simplification of
this contraction technique. Our result is as follows:
\begin{thm}
\label{thm:main}There is a deterministic algorithm that, given a
simple graph with $m$ edges, computes its edge connectivity in $m^{1+o(1)}$
time.\footnote{It is easy to extend the algorithm to compute the corresponding minimum
cut but we omit it here. }
\end{thm}

Expander decomposition is one of the most versatile tools in the
area of graph algorithms. Its existence was first exploited for graph
property testing \cite{GoldreichR02} and then for approximation algorithms
\cite{Trevisan05,ChekuriKS05,ChekuriKS13}. Fast algorithms of expander
decomposition \cite{SpielmanT13,OrecchiaV11,OrecchiaSV12,SaranurakW19,ChuzhoyGLNPS20}
are at the core of almost-linear time algorithms for many fundamental
problems including (directed) Laplacian solvers \cite{SpielmanT14,CohenKPPRSV17},
max flows \cite{KelnerLOS14}, matching \cite{BrandLNPSSSW20}, and various types of graph sparsifiers
\cite{SpielmanT11,ChuGPSSW18,Cheng2020sparsification,ChalermsookDLKPPSV20}
and sketchings \cite{AndoniCKQWZ16,JambulapatiS18}. More recently,
it has been used to break many long-standing barriers in the areas of dynamic
algorithms \cite{NanongkaiS17,Wulff-Nilsen17,NanongkaiSW17,ChuzhoyK19,ChuzhoyGLNPS20,BernsteinBGNSSS20,BernsteinGS20scc,GoranciRST20,JinS20,ChuzhoyS20}
and distributed algorithms \cite{EdenR18,ChangPZ19,DagaHNS19,ChangS19,ChangS20_det,Censor-HillelGL20}.

Unfortunately, how the expander decomposition has been applied is
usually highly non-trivial; it is either a step in a much bigger algorithm
containing other complicated components, or the guarantee of the decomposition
is exploited via involved analysis.

Both our algorithm and analysis are straightforward. The only key
step of the algorithm simply applies the expander decomposition followed
by the simple trimming and shaving procedures defined in
\cite{KawarabayashiT19}. We note that the idea of using expander
decomposition for edge connectivity actually appeared previously in
the distributed algorithm by \cite{DagaHNS19}.
%which the author coauthored.
However, that work requires many other distributed algorithmic components
and inevitably played down the simplicity of this approach. In fact, since of the original work by \cite{KawarabayashiT19}, their discussion in Sections 1.4 and 1.5 strongly suggested that expander decomposition should be useful. 
We hope
that this paper can highlight this simple idea and serve as a gentle introduction
on how to apply expander decomposition in general.

%{\color{blue} DW: I think K-T is basically using the expander decomposition idea already although they didn't say it that explicitly, but if you look at Section 1.4, 1.5 in https://arxiv.org/pdf/1411.5123.pdf , observation 5 plus the certify-or-cut procedure is basically it.}

The $m^{o(1)}$ factor in \Cref{thm:main} solely depends on quality
and efficiency of expander decomposition algorithms. It is believable
that this factor can be improved to $\polylog(n)$, which would immediately
improve the running time of our algorithm to $O(m\polylog(n))$.

\begin{comment}
how we simplify this last technique refer to \cite{Lo2020compact,RubinsteinSW18}
\end{comment}

\section{Preliminaries}

For any graph $G=(V,E)$ and a vertex set $S$, the \emph{volume}
of $S$ is denoted by $\vol_{G}(S)=\sum_{v\in S}\deg(v)$. For any
$A,B\subseteq V$, let $E(A,B)$ denote the set of edges with one
endpoint in $A$ and another in $B$. 
Let $\delta$ denote the minimum vertex degree of $G$.
Now, we state the key tool,
the expander decomposition. 
\begin{lem}
[Corollary 7.7 of \cite{ChuzhoyGLNPS20}]\label{lem:decomp}There is
an algorithm denoted by $\expander(G,\phi)$ that, given an $m$-edge
graph $G=(V,E)$ and a parameter $\phi \ge 0$, in $O(m\gamma)$
time where $\gamma=m^{o(1)}$, returns a partition $\X=\{X_{1},\dots,X_{k}\}$
of $V$ such that
\begin{itemize}
\item $\sum_{i}|E(X_{i},V\setminus X_{i})|=O(\phi m\gamma)$, and 
\item For each $i$ and each $\emptyset\neq S\subset X_{i}$, $|E(S,X_{i}\setminus S)|\ge\phi\min\{\vol_{G}(S),\vol_{G}(X_{i}\setminus S)\}$.\footnote{In \cite{ChuzhoyGLNPS20}, this guarantee is stated in a slightly
weaker form. This can be strengthen w.l.o.g.~(see
\Cref{sec:strong decomp}).}
\end{itemize}
\end{lem}
Note that if $\phi \ge 1/\gamma$, then the trivial partition $\X = \{v\mid v\in V\}$ satisfies the above guarantees.

The next tool is a deterministic algorithm by Gabow for computing
edge connectivity. Gabow's algorithm, in fact, can return the corresponding
minimum cut and also works for directed graphs, but we don't need
these guarantees in this paper.
\begin{lem}
[\cite{Gabow95}]\label{lem:Gabow}There is an algorithm that, given
an $m$-edge graph $G=(V,E)$ and a parameter $k$, in time $O(m\cdot\min\{\lambda,k\})$
returns $\min\{\lambda,k\}$ where $\lambda$ is the edge connectivity
of $G$.
\end{lem}

Lastly, we describe the $\trim$ and $\shave$ procedures
from \cite{KawarabayashiT19}. 
\begin{defn}
For any vertex set $S$ of a graph $G=(V,E)$, let $\trim(S)\subseteq S$
be obtained from $S$ as follows: while there exists a vertex $v\in S$ where
$|E(v,S)|<2\deg(v)/5$, removes $v$ from $S$. Let $\shave(S)=\{v\in S\mid|E(v,S)|>\deg(v)/2+1\}$.
\end{defn}
Note that, for every $v\in\trim(S)$,
$|E(v,\trim(S))|\ge2\deg(v)/5$.
%The \emph{trimming} procedure is given a vertex set $S$ in a graph $G=(V,E)$ and repeatedly removes from $S$ a vertex $v\in S$ such that $|E(v,S)|<2\deg(v)/5$ as long as such a vertex $v$ exists. Let $\trim(S)\subseteq S$ denote the resulting subset of $S$ after \emph{trimming}. Note that, for every $v\in\trim(S)$, $|E(v,\trim(S))|\ge2\deg(v)/5$. For any vertex set $S$ in $G$, let $\shave(S)=\{v\in S\mid|E(v,S)|>\deg(v)/2+1\}$ be the subset of $S$ after \emph{shaving}. 
Intuitively, the main difference between the two procedures is that $\trim$ keeps removing a vertex with low ``inside-degree'' as long as it exists, while $\shave$ removes all low ``inside-degree'' vertices once. 

\section{Algorithm and Analysis}

\begin{algorithm}
\begin{enumerate}
%\item \label{step:find delta}Compute the minimum vertex degree $\delta$ of $G$.
%\item \label{step:Gabow small delta}If $\delta\le50$, then \textbf{return} the edge connectivity $\lambda$ of $G$ using Gabow's algorithm (\Cref{lem:Gabow}).
\item \label{step:main}Compute $\X=\expander(G,40/\delta)$, $\X'=\{\trim(X)\mid X\in\X\}$,
 $\X''=\{\shave(X')\mid X'\in\X'\}$. 
\item \label{step:contract}Let $G'$ be the graph obtained from $G$ by
contracting every set $X''\in\X''$. 
\item \label{step:Gabow contracted}Using Gabow's algorithm (\Cref{lem:Gabow}),
\textbf{return }$\min\{\lambda',\delta\}$ where $\lambda'$ is the edge connectivity of $G'$ and $\delta$ is the minimum vertex degree of $G$.
\end{enumerate}
\caption{\label{alg:edgeconn}Computing edge connectivity $\lambda$ of a simple
graph $G$}
\end{algorithm}

%Let $\gamma=m^{o(1)}$ be the factor from \Cref{lem:decomp}. 
Our algorithm
is summarized in \Cref{alg:edgeconn}. 
Step~\ref{step:main} is \emph{the} step that simplifies the previous algorithms
by \cite{KawarabayashiT19,HenzingerRW17}. This step gives us a contracted
graph $G'$ of $G$ that preserves all non-trivial minimum cuts, as
will be proved in \Cref{lem:can contract} below. Previous algorithms
for computing such contraction are much more involved. For example,
they require an intricate analysis of PageRank \cite{KawarabayashiT19}
or local flow \cite{HenzingerRW17}. Moreover, both algorithms \cite{KawarabayashiT19,HenzingerRW17}
sequentially contract a part of $G$ into a supervertex and need to
distinguish supervertices and regular vertices thereafter. For us,
$G'$ is simply obtained by contracting each set $X''\in\X''$ simultaneously. 

Besides Step~\ref{step:main} of \Cref{alg:edgeconn} and the key lemma
below (\Cref{lem:can contract}), other algorithmic steps and analysis
follow the same template in \cite{KawarabayashiT19}. We only show
an alternative presentation for completeness.
\begin{lem}
\label{lem:can contract}Let $G = (V,E)$ be a simple graph. Let $(C,V\setminus C)$ be a non-trivial minimum cut in $G$. Let $X\in\expander(G,40/\delta)$, $X'=\trim(X)$, and $X''=\shave(X')$.
We have that
\begin{enumerate}
\item $\min\{|X\cap C|,|X\setminus C|\}\le\lambda/40$,
\item $\min\{|X'\cap C|,|X'\setminus C|\}\le2$, and
\item $\min\{|X''\cap C|,|X''\setminus C|\}=0$.
\end{enumerate}
In particular, $G'$ preserves all non-trivial minimum cuts of $G$.
\end{lem}

\begin{proof}
(1): We have $\min\{|X\cap C|,|X\setminus C|\}\le\lambda/40$ because of
the following:
\begin{align*}
\lambda & \ge|E(X\cap C,X\setminus C)| & \text{as }C\text{ is a minimum cut}\\
 & \ge (40/\delta)\cdot\min\{\vol_{G}(X\cap C),\vol_{G}(X\setminus C)\} & \text{by \Cref{lem:decomp}}\\
 & \ge 40\cdot\min\{|X\cap C|,|X\setminus C|\}.
\end{align*}

(2): Assume w.l.o.g.~that $|X'\cap C|\le|X'\setminus C|$. So, $|X' \cap C|  \le  \min\{|X \cap C|,|X \setminus C| \} \le \lambda/40$ by (1). Observe that 
\begin{align*}
\delta\ge\lambda & \ge|E(X'\cap C,X'\setminus C)| & \text{as }C\text{ is a minimum cut}\\
 & =\vol_{G[X']}(X'\cap C)-2|E(X'\cap C,X'\cap C)|\\
 & \ge\frac{2}{5}\delta|X'\cap C|-2|X'\cap C|^{2} & \text{as }X'=\trim(X)\text{ and }G\text{ is simple}.
\end{align*}
From the above, we conclude $|X'\cap C|\le2$. Otherwise, $|X'\cap C|\ge3$ and so $\delta \ge (6/5)\delta - 6|X' \cap C|$, which implies that $|X' \cap C| \ge \delta /30$. But we have $|X'\cap C|\le\lambda/40<\delta/30$, which is a contradiction.
%The above quadratic inequality implies that either $|X'\cap C|\le2$ or $|X'\cap C|>\delta/10$ (using the fact that $\delta>50$ and $|X'\cap C|$ is integral). But, since $|X'\cap C|\le\lambda/20<\delta/10$, we have $|X'\cap C|\le2$. 

(3): Again, assume w.l.o.g.~that $|X'\cap C|\le|X'\setminus C|$.
Suppose for contradiction that $\min\{|X''\cap C|,|X''\setminus C|\}>0$.
So there is a vertex $v\in X''\cap C\subseteq X'\cap C$. As $G$
is simple and $|X'\cap C|\le2$ by (2), we have $|E(v,X'\cap C)|\le1$.
Also, we have $|E(v,X')|>\deg(v)/2+1$ because $X''=\shave(X')$.
Therefore, $|E(v,X'\setminus C)|=|E(v,X')|-|E(v,X'\cap C)|>\deg(v)/2+1-1=\deg(v)/2.$
As $(C,V\setminus C)$ is non-trivial, we can switch $v$ from $C$
to $V\setminus C$ and obtain a smaller cut, contradicting the fact
that $C$ is a minimum cut.
\end{proof}
\begin{cor}
\label{cor:correct}\Cref{alg:edgeconn} correctly computes the edge
connectivity $\lambda$ of $G$.
\end{cor}

\begin{proof}
%Consider \Cref{alg:edgeconn}. 
%We can assume that $\delta>50$, otherwise \Cref{step:Gabow small delta} would correctly return $\lambda$. 
Note that $\lambda'\ge\lambda$ because $G'$ is obtained from $G$ by
contraction. If $\lambda=\delta$ (i.e.~there is a trivial minimum
cut), then $\min\{\lambda',\delta\}=\lambda$. If $\lambda<\delta$
(i.e.~all minimum cuts are non-trivial), then we have $\lambda'=\lambda$
by \Cref{lem:can contract} and so $\min\{\lambda',\delta\}=\lambda$.
\end{proof}
\begin{lem}
\label{lem:bound G'}The contracted graph $G'$ has at most $O(m\gamma/\delta)$
edges.
\end{lem}

\begin{proof}
Assume that $\delta\ge4$ otherwise the statement is trivial. Let
$G/\X$ denote the graph obtained from $G$ by contracting each $X\in\X$
into a single vertex. Let $G/\X'$ and $G/\X''$ be similarly defined.
Note that $G'=G/\X''$. We would like to bound $|E(G/\X'')|=|E(G/\X)|+|E(G/\X')\setminus E(G/\X)|+|E(G/\X'')\setminus E(G/\X')|$.
We will show that each term is bounded by $O(m\gamma/\delta)$ where $\gamma$ is the factor from \Cref{lem:decomp}.

First, the set $E(G/\X)$ contains exactly the edges crossing the
partition $\X$ of $V$. So $|E(G/\X)|=\frac{1}{2}\sum_{X\in\X}|E(X,V\setminus X)|=O(m\gamma/\delta)$
by \Cref{lem:decomp}.

Second, the set $E(G/\X')\setminus E(G/\X)$ contains all edges that
are ``trimmed from'' each $X\in\X$. Consider the $\trim$ procedure
executing on $X$ until $X$ becomes $X'$. Whenever a vertex $v$
is removed from $X$, $|E(X,V\setminus X)|$ is decreased by at least
$\deg(v)/5$, because at that point of time $|E(v,X)|\le2\deg(v)/5$
but $|E(v,V\setminus X)|\ge3\deg(v)/5$. On the other hand, the number
of trimmed edges, $E(G/\X')\setminus E(G/\X)$, is increased by at
most $|E(v,X)|\le2\deg(v)/5$. Initially, we have $\sum_{X\in\X}|E(X,V\setminus X)|=2|E(G/\X)|$.
As we argued, every two units in $|E(G/\X')\setminus E(G/\X)|$ can
be charged to one unit in $\sum_{X\in\X}|E(X,V\setminus X)|$. So
$|E(G/\X')\setminus E(G/\X)|\le4|E(G/\X)|=O(m\gamma/\delta)$.

Last, the set $E(G/\X'')\setminus E(G/\X')$ contains all edges that
are ``shaved from'' each $X'\in\X'$. The number of shaved edges
from $X'$ is bounded by $\sum_{v\in X'\setminus\shave(X')}|E(v,X')|$.
By definition of $\shave$, for each vertex $v\in X'\setminus\shave(X')$,
we have $|E(v,X')|<\deg(v)/2+1$ and so $|E(v,V\setminus X')|>\deg(v)/2-1$.
As $\delta\ge4$, we have $|E(v,X')|<4|E(v,V\setminus X')|$ and so
$\sum_{v\in X'\setminus\shave(X')}|E(v,X')|\le4|E(X',V\setminus X')|$.
Summing over all $X'\in\X'$, we have $|E(G/\X'')\setminus E(G/\X')|\le4\sum_{X'\in\X'}|E(X',V\setminus X')|\le4\sum_{X\in\X}|E(X,V\setminus X)|=O(m\gamma/\delta)$.
The last inequality is because the $\trim$ procedure only decreases
$|E(X,V\setminus X)|$ and so $|E(X',V\setminus X')|\le|E(X,V\setminus X)|$
for each $X'=\trim(X)$. 
\end{proof}
\begin{cor}
\label{cor:time}\Cref{alg:edgeconn} takes $O(m\gamma)=m^{1+o(1)}$
time.
\end{cor}

\begin{proof}
%Computing the minimum degree $\delta$ in Step~\ref{step:find delta} takes $O(m)$ time. If $\delta\le50$, then Gabow's algorithms from \Cref{lem:Gabow} takes $O(m\delta)=O(m)$ in Step~\ref{step:Gabow small delta}. 
In Step~\ref{step:main}, $\X$ can be computed in $O(m\gamma)$ time by \Cref{lem:decomp}.
%Note that we can apply \Cref{lem:decomp} because the parameter $\phi=20/\delta\le1$.
$\X'$ and $\X''$ can be computed in $O(m)$ by using straightforward
implementations for $\trim$ and $\shave$. Contracting $G$ into $G'$ can be done in $O(m)$ time in Step~\ref{step:contract}. Finally, in Step~\ref{step:Gabow contracted},
the minimum degree $\delta$ can be computed in $O(m)$ time, and Gabow's algorithm takes $O(|E(G')|\delta)=O(m\gamma)$ time by \Cref{lem:bound G'}. 
\end{proof}
To conclude, \Cref{thm:main} follows immediately from Corollaries \ref{cor:correct}
and \ref{cor:time}.

\section*{Acknowledgements}
I thank Aaron Bernstein for encouragement for writing this note up. 
Also, thanks to Sayan Bhattacharya, Maximilian Probst Gutenberg, Jason Li, Danupon Nanongkai, and Di Wang for helpful comments on the write-up.

\bibliographystyle{alpha}

\bibliography{references}

\appendix

\section{Variants of Expander Decomposition}

\label{sec:strong decomp}

The guarantee for expander decomposition is usually stated in a weaker
form: for every $\emptyset\neq S\subset X_{i}$ and $i$, we have
$|E(S,X_{i}\setminus S)|\ge\phi\min\{\vol_{G[X_{i}]}(S),\vol_{G[X_{i}]}(X_{i}\setminus S)\}$
instead of $|E(S,X_{i}\setminus S)|\ge\phi\min\{\vol_{G}(S),\vol_{G}(X_{i}\setminus S)\}$
as in \Cref{lem:decomp}. In \cite{ChuzhoyGLNPS20}, they also stated
the guarantee in this weaker form.

Here, we argue that the stronger form can be assumed without loss of
generality. This observation already appeared in \cite{ChuzhoyS20}.
Let $G=(V,E)$ be any $m$-edge graph and let $G'$ be
obtained from $G$ by adding $\deg_{G}(v)$ self-loops to each vertex
$v$. So $G'$ has $m'=O(m)$ edges. Suppose we have obtained a weaker
form of expander decomposition $\X=\{X_{1},\dots,X_{k}\}$ of $G'$.
That is, $\sum_{i}|E_{G'}(X_{i},V\setminus X_{i})|=O(\phi m'\gamma)$
and $|E_{G'}(S,X_{i}\setminus S)|\ge\phi\min\{\vol_{G'[X_{i}]}(S),\vol_{G'[X_{i}]}(X_{i}\setminus S)\}$
for every $\emptyset\neq S\subset X_{i}$ and $i$. 

Observe that $E_{G}(A,B)=E_{G'}(A,B)$ for any two disjoint sets $A,B\subseteq V$.
So $\sum_{i}|E_{G}(X_{i},V\setminus X_{i})|=O(\phi m'\gamma)=O(\phi m\gamma)$.
Also, we have 
\[
|E_{G}(S,X_{i}\setminus S)|=|E_{G'}(S,X_{i}\setminus S)|\ge\phi\min\{\vol_{G'[X_{i}]}(S),\vol_{G'[X_{i}]}(X_{i}\setminus S)\}\ge\phi\min\{\vol_{G}(S),\vol_{G}(X_{i}\setminus S)\}
\]
where the last inequality is because of the self-loops in $G'$. That
is, $\X$ is indeed a stronger form of expander decomposition of $G$
(modulo losing a constant factor in the bound of $\sum_{i}|E_{G'}(X_{i},V\setminus X_{i})|$). 
\end{document}